\documentclass[11pt]{llncs}
\usepackage{xspace}
\usepackage{amsmath}
\usepackage{amssymb}
\usepackage{amsfonts}
\usepackage{fullpage,times}

\newtheorem{observation}[theorem]{Observation}

\newcommand{\algo}[1]{\ensuremath{\mbox{\textsc{#1}}}\xspace}

\newcommand{\alg}{\algo{Alg}}
\newcommand{\alga}{\algo{A}}
\newcommand{\opt}{\algo{opt}}
\newcommand{\optn}{\algo{opt-nice}}
\newcommand{\sol}{\algo{sol}}
\newcommand{\apx}{\algo{apx}}
\newcommand{\soln}{\algo{sol-nice}}

\newcommand{\pr}{\mbox{{\sc{nebc}}}}
\newcommand{\pra}{\mbox{{\sc{mep}}}}
\newcommand{\prb}{\mbox{{$\Delta$-\sc{nebc}}}}

\newcommand{\I}{{\cal I}}
\newcommand{\C}{{\cal C}}

\newcommand{\h}{{\cal H}}
\newcommand{\B}{{\cal B}}

\newcommand{\Z}{{\mathbb{Z}}}

\newcommand{\eps}{\varepsilon}

\begin{document}
%%%\begin{titlepage}

\title{The near exact bin covering problem\thanks{This research was supported by a grant from the ISF, the Israel Science Foundation (grant number 308/18).}}

\author{Asaf Levin\inst{1}}
\institute{Faculty of
Industrial Engineering and Management, The Technion, 32000 Haifa,
Israel. \email{levinas@ie.technion.ac.il}. }\date{}

\maketitle

%%\thispagestyle{empty}
%%
%%\medskip
%%\medskip

\begin{abstract} 
We present a new generalization of the bin covering problem that is known to be a strongly NP-hard problem.  In our generalization there is a positive constant $\Delta$, and we are given a set of items each of which has a positive size.  We would like to find a partition of the items into bins.  We say that a bin is near exact covered if the total size of items packed into the bin is between $1$ and $1+\Delta$. Our goal is to maximize the number of near exact covered bins. If $\Delta=0$ or $\Delta>0$ is given as part of the input, our problem is shown here to have no approximation algorithm with a bounded asymptotic approximation ratio (assuming that $P\neq NP$). However, for the case where $\Delta>0$ is seen as a constant, we present an asymptotic fully polynomial time approximation scheme (AFPTAS) that is our main contribution.
\end{abstract}

%%%% \end{titlepage}
\section{introduction}
We are given a  parameter $\Delta>0$ (independent of the input).  The input to the {\sc near exact bin covering problem} (\pr) consists of a set of $n$ input items $\I=\{1,2,\ldots ,n\}$ where item $j$ is associated with its {\em size} $s_j\in (0,1]$.  A feasible solution is a partition of the items into subsets called bins.  For a given bin (in a fixed solution) we say that the bin is {\em near exact covered} if the total size of items in the bin is at least $1$ and strictly smaller than $1+\Delta$.  The goal function (also known as reward function or objective function) is the number of near exact covered bins.  Our problem \pr\ is to find a partition of the item set into bins so as to maximize the goal function. We refer to the variant of \pr\ where $\Delta$ is given as part of the input  as \prb.

We focus on the offline settings of the problem and study it with respect to the criterion of asymptotic approximation ratio defined as follows.  Given an algorithm $\alga$ and an input $I$ we denote by $\alga(I)$ the goal function value of the solution output by $\alga$. We denote by $\opt$ an (exponential-time) algorithm that computes an optimal solution for our problem.  We say that a polynomial time algorithm $\alg$ is a $\rho$ asymptotic approximation algorithm for \pr\ if $$\limsup_{I:\opt(I)\rightarrow \infty} \frac {\opt(I)}{\alg(I)} \leq \rho \ .$$  In many cases it is easier to establish a stronger upper bound on the reward of the algorithm, and it is well-known that if there exists a constant $C$ such that for every instance $I$ we have $$\opt(I)\leq \rho\cdot \alg(I)+C , $$ then $\alg$ is a $\rho$ asymptotic approximation algorithm.  We  refer to the value of $\rho$ as the asymptotic approximation ratio  of algorithm $\alg$.   An {\em asymptotic polynomial time approximation scheme} (APTAS) is a family of algorithms such that for every $\eps>0$, the family contains an algorithm $\alg_{\eps}$ that is an $1+\eps$ asymptotic approximation algorithm.  In this family the constant $C$ in the stronger definition of asymptotic approximation ratio may depend on $\eps$.  If we require that the time complexity upper bound on $\alg_{\eps}$ is polynomial of the input encoding length and $1/\eps$, then the family is referred to as an {\em asymptotic fully polynomial time approximation scheme} (AFPTAS).

Next, we describe the connection of \pr\ to the bin covering problem that was studied in the literature.
The special case of \pr\ where $\Delta=1$ is equivalent to the {\em bin covering problem}.  This equivalence means that a solution to one problem can be transformed to the other problem and vice versa.  The definition of bin covering allows bins of total size larger than $2$ (and they are still considered to be covered) but since the reward for such a bin is $1$, we can delete one item at a time (to be packed in a dedicated bin) from such a bin and create a solution to \pr\ with at least the same value.  On the other hand, a solution to \pr\ is a feasible solution for the bin covering problem with at least the same reward.  The bin covering problem was suggested in \cite{Ass1,Ass2}. They proved that
the greedy algorithm (which simply keeps putting items into the
same bin until it is covered and then moves on to the next bin)
has an asymptotic approximation ratio of $2$. This is best possible for online algorithms as shown by \cite{CT}. Moreover, two more
offline algorithms were derived with asymptotic
approximation ratios $\frac 32$ and $\frac43$, respectively.  Most relevant to our work
Csirik, Johnson, and Kenyon \cite{CJK01} designed an
elegant APTAS for bin covering.  The running time of this asymptotic scheme was improved into an AFPTAS by Jansen and Solis-Oba \cite{JS03}.  Since \pr\ is a generalization of the bin covering problem, our AFPTAS for \pr\ will use some of the techniques of \cite{CJK01,JS03}.  Additional algorithmic results established for the bin covering problem and variants of it appear in e.g. \cite{CF90,CFGR,WZ99,Ep01,EIL10,HS12,CFL14,Ja16,FR16,FR18,BEL18,BEJLMR19}.

For the bin covering problem we have the following motivating application. Assume that we have a food producer that sells its product as containers.  The product has a minimum weight per container, and the goal of the producer is to maximize the number of packed containers that meet this minimum weight constraint.  With this application in mind, the bin covering problem is the case where there is no upper limit on the weight of a container.  Obviously, the producer has certain packages used for packing the product that cannot be over-packed.  This requires an upper bound on the total weight of the packed containers, so we get an instance of \pr. 

In the bin packing literature where the total size of items packed into a bin must be at most $1$ we refer to \cite{FerLue81,KarKar82} for an AFPTAS. The variant where the cost of the bin depends on the total size of items packed into the bin, is called {\em bin packing with bin utilization cost}.  For this variant under the condition that the cost function is monotone non-decreasing there is an AFPTAS established in \cite{EL17}.  Note that in \pr\ the reward of a bin is not a monotone function of the total size of items packed in the bin.

\paragraph{Paper outline.} In Section \ref{sec:pre} we consider a variant of \pr\ where $\Delta=0$, and show that this variant does not admit an approximation algorithm with a bounded asymptotic approximation ratio.  This result also implies the same hardness of approximation for \prb\ as we show in Section \ref{sec:pre}.  Then, we turn our attention to designing an AFPTAS for \pr\ that is our main contribution.  We start our exposition of this result in Section \ref{sec:afptas} where we discuss some initial steps and mainly describe the main guessing step allowing the algorithm to partition the instance into two subproblems that can be solved independently.  Then, in Section \ref{I1sec} we approximate the first subproblem, and in Section \ref{I2sec} we approximate the second subproblem. The novelty of our scheme is mainly in the guessing step allowing the algorithm to partition the problem into two independent subproblems.

\section{Hardness of approximation of \prb \label{sec:pre}}
We define the {\sc Maximum Exact Partition} problem denoted as \pra. Specifically, \pra\  is the variant of \pr\ with $\Delta=0$.  Observe that the standard reduction from $3$-partition to bin packing implies that \pra\ is NP-hard in the strong sense.  Furthermore, the standard reduction from partition implies that in \pra\ it is NP-hard to distinguish between instances in which the optimal value is at least $2$ and instances in which the optimal value is zero.  Note that in instances of \pra\ where all items have sizes larger than $\frac 1{t+1}$ (for a constant value of $t$) there is a $\frac{t}{2}+\eps$ approximation algorithm for every $\eps>0$ by \cite{HurS}. Thus, the proof of our next inapproximability result needs to be based on instances where some of the items are small.

\begin{theorem}
Unless $P=NP$, for every constant value of $\rho$, problem \pra\ does not admit a polynomial time algorithm $\alg$ with an asymptotic approximation ratio that is at most $\rho$.
\end{theorem}
\begin{proof}
Assume by contradiction that the claim does not hold for a constant $\rho$. So there is a polynomial time algorithm $\alg$ for \pra\ with an asymptotic approximation ratio that is not larger than $\rho$.  Without loss of generality, $\rho$ is an integer (otherwise, we can increase $\rho$ to an integer).  By definition of $\limsup$, we know that there is a positive integer $T$ such that if $\opt(I)\geq T$, then $\alg(I) \geq \frac{\opt(I)}{\rho+1}\geq \frac{T}{\rho+1}$, whereas if $\opt(I)=0$, then $\alg(I)=0$.   Note that $T$ is a constant defined by the performance guarantee of $\alg$. Thus it suffices to show that deciding if $\opt(I)$ is exactly zero is NP-complete even if we assume that if $\opt(I)>0$ then $\opt(I)\geq T$.  We show this claim via reduction from the partition problem.

Let $a_1,a_2,\ldots ,a_n$ be positive integers such that $\sum_{i=1}^n a_i =2B$ where $B$ is an integer.  The partition problem asks if there is a subset $S\subseteq \{ 1,2,\ldots ,n\}$ such that $\sum_{i\in S} a_i=B$.  The partition problem is known to be NP-complete \cite{GJ}.  Given such an instance to the partition problem we consider the following instance to \pra.

For $j=1,2,\ldots ,T$, we have $n+1$ items which will be referred to as {\em generation $j$ items}.  The first $n$ items of generation $j$ have sizes $$a_{i,j}=\frac{a_i}{(3B)^j} \ \ , \  i=1,2,\ldots ,n$$ and the last item of generation $j$ has size $$a_{n+1,j}=1-\frac{B}{(3B)^j} \ .$$ We say that item $(i,j)$ is the $i$-th item of generation $j$ (the one of size $a_{i,j}$).  This defines the instance of \pra.  Note that $a_{i,j}>1/2$ if and only if $i=n+1$ (for all $j$).  We argue next the following two claims. The first claim is that if the partition instance is a YES instance, then the optimal solution value of the instance of \pra\ is at least $T$. The second claim is that if the partition instance is a NO instance, then every feasible solution of the instance of \pra\ has a zero reward.

Assume first that the partition instance is a YES instance. Let $S \subseteq \{ 1,2,\ldots,n\}$ be an index subset such that $\sum_{i\in S} a_i=B$.  For every generation $j$, the items of generation $j$ are packed as follows.  We have one {\em exact bin} of generation $j$ with the items $(i,j)$ for all $i\in S\cup\{ n+1\}$.  The other items of generation $j$ are packed into dedicated bins (one bin per item).  The total size of the items in the exact bin of generation $j$ is $$\sum_{i\in S\cup \{ n+1\}}a_{i,j}= \sum_{i\in S} \frac{a_i}{(3B)^j} + 1-\frac{B}{(3B)^j} = 1$$ where the last equality holds as $\sum_{i\in S} a_i=B$. So indeed the objective function value of our solution for \pra\ is at least $T$.   

Next, assume that the partition instance is a NO instance. Assume by contradiction that there is a feasible solution with at least one bin denoted as $\B$ whose items have total size exactly $1$.  Consider the set of items packed into $\B$ then by definition it has at most one item of size larger than $1/2$.  However,  since $\sum_{j=1}^T \sum_{i=1}^n a_{i,j}<1$, $\B$ must have exactly one item of size larger than $1/2$.  Assume that it is item $(n+1,j)$, that is, the last item of generation $j$.

We claim that all items packed in $\B$ are of generation $j$.  First assume that there is an item $(i,j')$ of generation $j'<j$ packed into $\B$, then $a_{i,j'}\geq \frac{1}{(3B)^{j'}}> 1-a_{n+1,j}$. So the total size of items in $\B$ is strictly larger than $1$ that contradicts our assumption on $\B$.  Thus, all items packed into $\B$ are of generation at least $j$.  Next, observe that all items of generations strictly larger than $j$ which are smaller than $1/2$ have total size  $$\sum_{j''=j+1}^T\sum_{i=1}^n a_{i,j''}= \sum_{j''=j+1}^T \frac{2B}{(3B)^{j''}}<\frac{1}{(3B)^j} \ .$$ However, the size of every item of generation $j$ is an integer multiple of $\frac{1}{(3B)^j}$.  Thus,  $\B$ must contain items of a common generation $j$.  We let $$S=\{ i: (i,j)\mbox{ is packed into bin } \B \} \setminus \{ n+1\} \ . $$  By definition of $\B$ and $S$, we conclude that $\sum_{i\in S} a_i= (3B)^j \sum_{i\in S} a_{i,j} =B$. Therefore, the partition instance is a YES instance after all, contradicting our assumption.  Therefore, such a bin $\B$ may not exist.
\qed\end{proof}

Observe that in the last proof, if we use $\Delta=\frac{1}{(3B)^{T+1}}$, we get an instance of \prb\ where the binary encoding length of $\Delta$ is upper bounded by a polynomial in the binary encoding length of the items in this instance. In the resulting instance of \prb\ whenever a near exact covered bin exists it must have items of total size exactly $1$.  This property holds as all items have sizes that are integer multiples of $\frac{1}{(3B)^{T}}$. Therefore, the last hardness of approximation claim holds even for \prb.  We summarize this conclusion in the following theorem.

\begin{theorem}
Unless $P=NP$, for every constant value of $\rho$, problem \prb\ does not admit a polynomial time algorithm with an asymptotic approximation ratio that is at most $\rho$.
\end{theorem}

\section{The initial steps of the AFPTAS for \pr\label{sec:afptas}}
Let $\eps>0$ be such that we would like to get an algorithm with an asymptotic approximation ratio $(1+\eps)^c$ for a constant $c$ for \pr\ whose time complexity is upper bounded by a polynomial in the input encoding length and in $\frac 1{\eps}$.  Without loss of generality we assume that $\eps\leq \frac{1}{12}$ and that $\frac{1}{\eps}$ is an integer. We define $\delta>0$ to be the largest value satisfying $\delta \leq  \frac{\Delta}{4}$ for which $\frac{1}{\delta}$ is an integer.  Then $\delta\geq \frac{\Delta}{8}$ and since $\Delta$ is a positive constant, so does $\delta$.

Our scheme has an initial item classification step followed by a characterization of near optimal solutions.  This characterization motivates a guessing step that basically separates the input into two parts.  One part of the input  is handled using the methods of \cite{FerLue81,KarKar82} developed for the bin packing problem (see Section \ref{I1sec}). Whereas the second part of the input is tackled using the methods of \cite{CJK01,JS03} for the bin covering problem (see Section \ref{I2sec}).  The novelty of our approach lies in the characterization of the near optimal solution together with the resulting guessing step.  We turn our attention to the description of the necessary background for presenting the characterization of the optimal solution.

\subsection{The initial classification of items into huge and non-huge items}
We say that an item $j$ whose size is $s_j$ is a {\em huge item} if $s_j\geq \delta$, and otherwise it is a {\em non-huge item}.  The set of huge items is denoted as $\h$, the set of all items is denoted as $\I$, so the set of non-huge items is $\I \setminus \h$.

We partition $\h$ into classes based on the following rule.  The item set $\h_{\psi}$ (for $\psi\in \{0,1,\ldots, \frac{1}{\delta^3}-\frac{1}{\delta^2}$) is defined as $$\h_{\psi}=\{ j\in \h: \delta + \psi\cdot \delta^3\leq s_j < \delta+ (\psi+1)\cdot \delta^3\} \ .$$ Observe that this is indeed a partition of $\h$, and we say that $\h_{\psi}$ is {\em class $\psi$ of huge items} that consists of $|\h_{\psi}|$ items.  The index set of classes is denoted as $\Psi$.  Furthermore, we assume that each such class is sorted in a non-decreasing order of sizes of items in this class breaking ties based on decreasing indexes. For example, when we say the $5$-th item in the class we refer to this sorting.

To motivate this classification, consider a bin $\B$ whose items have total size between $1$ and $1+\Delta$.  Consider the operation of replacing its huge items by another set of huge items such that for every class of huge items, the new set has the same number of items as $\B$ used to have in the original solution.  After applying this operation, the total size of items in $\B$ changes by at most $2\delta^2$.  This holds as $\B$ has at most $\frac{2}{\delta}$ huge items and each of which is replaced by an item whose size differ by at most $\delta^3$.  Furthermore, the number of classes of huge items is a constant.

\subsection{Characterization of nice solutions}

Next, we would like to show that every feasible solution $\sol$ for \pr\ can be transformed into a new solution  $\soln$ satisfying the following.  $\soln$ has an objective function value not significantly smaller than the objective function value of $\sol$ and $\soln$ has the characterization named being a {\em nice solution} with a {\em certificate vector}.  This characterization enables the next guessing step  of our scheme. 

\begin{definition}
A solution $\soln$ for \pr\ is a nice solution with a certificate vector $(v_{\psi},u_{\psi})_{\psi\in \Psi}$ if there is a partition of the near exact covered bins in $\soln$ into two sets called type $1$ and type $2$ bins such that the following properties hold.
\begin{enumerate}
\item\label{nice1} Every type 1 bin $\B$ in $\soln$ has items of total size at least $1+\delta$ and less than $1+\Delta$.  Every type 2 bin $\B$ in $\soln$ has items of total size at least $1$ and at most $1+3\delta$.
\item\label{nice2} Type 1 bins in $\soln$ do not contain non-huge items.
\item\label{nice3} For  every class $\psi\in \Psi$ the following holds. For every type 1 bin $\B$, the bin $\B$ may contain an item of the class only if the item is among the last $v_{\psi}$ items of the class. Similarly, every type 2 bin may contain an item of the class only if the item is among the first $u_{\psi}$ items of the class.
\item\label{nice4} For every $\psi\in \Psi$, we have $|\h_{\psi}| \geq u_{\psi}+v_{\psi}$. Moreover, both $u_{\psi}$ and $v_{\psi}$ are integer powers of $1+\eps\delta^3$ rounded down to the next integer, that is, $$u_{\psi},v_{\psi} \in \left\{ \left\lfloor \left( 1+\eps\cdot \delta^3\right)^t \right\rfloor : t\in \Z\right\},\ \ \forall \psi\in \Psi .$$
\item\label{nice5} Last, for every $\psi\in \Psi$, the following two conditions are satisfied. The number of items of $\h_{\psi}$ packed into near exact covered bins in $\soln$ of type $1$ is at least $\frac{v_{\psi}}{2}$. Furthermore, the number of items of the class packed into near exact covered bins of type $2$ in $\soln$ is at least $\frac{u_{\psi}}{2}$.
\end{enumerate}
\end{definition}

\begin{lemma}
Given a feasible solution $\sol$, there exists another feasible solution $\soln$ with some certificate vector $(v_{\psi},u_{\psi})_{\psi\in \Psi}$ that is a nice solution whose objective function value is at least $(1-2\eps)$ times the objective function value of $\sol$.
\end{lemma}
\begin{proof}
Consider $\sol$.  For every  bin $\B$ of $\sol$ that is not near exact covered, we pack every item of the bin in a dedicated bin. This operation does not decrease the objective function value of the solution.  Next, if a bin $\B$ of the resulting solution is near exact covered, but has items of total size at least $1+\delta$, we can assume without loss of generality that this bin has only huge items.  This is so as otherwise we remove one non-huge item at a time and pack this item into a dedicated bin until the first point in time where either the total size of the remaining items in $\B$ is smaller than $1+\delta$, or there are no further non-huge items in $\B$.  We apply this operation as long as there are such near exact covered bins.  This process does not decrease the number of near exact covered bins and create a solution satisfying the required property.  Let $\sol'$ be the resulting solution.

Next we partition the set of near exact covered bins in $\sol'$.  We say that a near exact covered bin $\B$ is a type $1$ bin if the total size of its items is at least $1+2\delta$ and otherwise it is a type $2$ bin.  Note that at this point $\sol'$ satisfies the first two properties of nice solutions.  In order to obtain a nice solution we apply the following procedure for every class $\psi$ of huge items.

For every $\psi$, we count the number of huge items of class $\psi$ that $\sol'$ packs into bins of type $1$ denoted as $\hat{v}_{\psi}$ and the number of such items that $\sol'$ packs into bins of type $2$ denoted as $\hat{u}_{\psi}$.  Our next step is to form an initial certificate vector $(v'_{\psi},u'_{\psi})_{\psi\in \Psi}$. This initial certificate vector will be modified later to the certificate vector that we will use for the proof of the lemma.  The value of $v'_{\psi}$ ($u'_{\psi}$, respectively) is the largest value in the set $\left\{ \left\lfloor \left( 1+\eps \cdot \delta^3\right)^t \right\rfloor : t\in \Z\right\}$ that is at most $\hat{v}_{\psi}$ (at most $\hat{u}_{\psi}$, respectively).  Note that $\hat{v}_{\psi}+\hat{u}_{\psi} \leq |\h_{\psi}|$.  By definition, $v'_{\psi}\leq \hat{v}_{\psi}$ and $u'_{\psi}\leq \hat{u}_{\psi}$ so $v'_{\psi}+u'_{\psi} \leq |\h_{\psi}|$ holds for all $\psi\in \Psi$. Our certificate vector will be component-wise not larger than the initial certificate vector so this required property will be satisfied.  

We apply the following process in which we process every class $\h_{\psi}$ of huge items.  In this process we say we delete a near exact covered bin and we mean that all its items are packed into dedicated bins and the number of near exact covered bins is decreased by $1$.  If $v'_{\psi}< \hat{v}_{\psi}$ we delete up to  $\eps\delta^3$ times the number of bins of type $1$ in $\sol'$ containing items of this class where we delete one such bin at a time until the first time when  the number of items in such non-deleted bins decreases to be at most $v'_{\psi}$ where the bins are sorted in non-increasing order of items of this class.   Similarly, if $u'_{\psi}< \hat{u}_{\psi}$, we delete up to  $\eps\delta^3$ times the number of bins of type $2$ in $\sol'$ containing items of this class where we delete one such bin at a time until the first time when the number of items of $\h_{\psi}$ in such non-deleted bins decreases to be at most $u'_{\psi}$ where the bins are sorted in non-increasing order of items of this class. Then, we move to the next class of huge items.  Note that at later iterations we may delete additional bins containing items of the class. We denote by $\sol''$ the resulting solution. Observe that whenever we process a class, we may delete up to $2\eps\delta^3$ times the number of near exact covered bins in $\sol'$. So in total we may decrease the objective function value by at most $2\eps$ times the objective function value of $\sol'$.  In later steps we will not decrease the number of near exact covered bins so the output solution will have an objective function value that is at least $(1-2\eps)$ times the objective function value of $\sol$.

At the end of this step, we set $v_{\psi}$ ($u_{\psi}$) to be the smallest value in the set $\left\{ \left\lfloor \left( 1+\eps \cdot\delta^3\right)^t \right\rfloor : t\in \Z\right\}$ that is at least the number of items of class $\psi$ contained in (remaining) bins of type $1$ (and $2$, respectively).  Observe that this last step cannot increase the values of the certificate vector, so properties 4,5 in the definition of nice solutions will be satisfied using this certificate vector.

Then we identify the set $S(\psi,1)$ of $v_{\psi}$ last items of class $\psi$. Whenever a type $1$ bin contains items of this class that do not belong to $S(\psi,1)$, we remove such items from the bin leaving space for items of class $\psi$ and add items from $S(\psi,1)$ to this bin (among the items which are not used to be packed in near exact covered bins of type $1$).  Similarly, we let $S(\psi,2)$ be the set of $u_{\psi}$ first items of class $\psi$. Whenever a type $2$ bin contains items of this class that do not belong to $S(\psi,2)$, we remove such items from the bin leaving space for items of class $\psi$ and add items from $S(\psi,2)$ to this bin (among the items which are not used to be packed in near exact covered bins of type $2$).  In both cases the number of items of each class in every bin is left without modification.

This process does not hurt the properties 2,4,5. Property 3 is trivially satisfied and it remains to take care of the first property.  We denote by $\soln$ the resulting solution and we consider the first property.  Before this last step occurred every bin of type $1$ had items of total size at least $1+2\delta$ and less than $1+\Delta$, and we replace up to $\frac{2}{\delta}$ huge items by other huge items of a common class.  Each such replacement can only decrease the total size of items in such bin and such a decrease is not larger than $\delta^3$.  Therefore, the total size of items in such bin is at least $1+2\delta-2\delta^2 > 1+\delta$ and at most $1+\Delta$.  Similarly every bin of type $2$   had items of total size at most $1+2\delta$ and at least $1$, and we replace up to $\frac{2}{\delta}$ huge items by other huge items of a common class.  Each such replacement can only increase the total size of items in such bin and such an increase is not larger than $\delta^3$.  Therefore, the total size of items in such bin is at least $1$ and at most $1+2\delta+2\delta^2 < 1+3\delta$.  So the first property of nice solutions is satisfied as well.
\qed\end{proof}

\subsection{Guessing the certificate vector and partitioning the input into independent problems}

We fix a nice solution whose objective function value is maximized.  We let $\optn$ be this solution.  Our guessing step is to guess the certificate vector $(u_{\psi},v_{\psi})_{\psi\in \Psi}$ corresponding to $\optn$ (or one such certificate vector if its identity is not well-defined).  By the term guess, we mean that the next steps of our scheme will be applied for every possible value of the guessed certificate vector, each of which will lead to a feasible solution, and at the end of the scheme we will output the solution with the maximum objective function value breaking ties arbitrarily.  In the analysis of our scheme we analyze the iteration of this exhaustive enumeration in which we have used the vector corresponding to $\optn$.  We next show that this exhaustive enumeration has a polynomial number of iterations for fixed $\Delta$.

\begin{lemma}
The number of iterations of the exhaustive enumeration loop implied by the guessing step is $O\left( (\frac{\log n}{\eps})^{O(1/\delta^3)} \right)$. 
\end{lemma}  
\begin{proof}
We have that the number of components in the certificate vector is $O(\frac{1}{\delta^3})$. Each such component is a non-negative integer that is at most $\max_{\psi} |\h_{\psi}| \leq n$ and it is a rounded value of an integer power of $1+\eps\cdot \delta^3$.  Therefore, the number of possible values for each component is at most $2+\log_{1+\eps\cdot \delta^3} n = O(\frac{\log n}{\eps}) $ using $\frac{\Delta}{8} < \delta<\Delta\leq 1$ and by the fact that $\Delta$ is a constant.  Therefore, the number of possible certificate vectors is $O\left( (\frac{\log n}{\eps})^{O(1/\delta^3)} \right)$.  
\qed\end{proof}

We define two disjoint subsets of items.  The first subset of items consists of items that may belong to bins of type 1 in $\optn$ and the second subset consists of items that may belong to bins of type 2 in $\optn$.  The first subset has for every $\psi\in \Psi$ the last $v_{\psi}$ items of the class (and this subset has no non-huge items).  The second subset consists of all non-huge items and the first $u_{\psi}$ items of every class $\psi$ of huge items.  We denote by $I_1$ and $I_2$ the two inputs to our problem where in $I_k$ the items are of the $k$-th subset.  By the fourth property in the definition of nice solutions we conclude that the two subsets are disjoint.  Our algorithm will find a feasible solution $\apx_1$ for $I_1$ and a feasible solution $\apx_2$ for $I_2$.  Then, the output of our scheme is the union of the bin sets of these two solutions and its objective function value is the sum of the objective function value of $\apx_1$ and the objective function value of $\apx_2$.

For the sake of our analysis, for $I_1$ we consider a solution that maximizes the number of bins with items of total size in $[1+\delta,1+\Delta)$ subject to the constraint that for every $\psi\in \Psi$ at least $\frac{v_{\psi}}{2}$ items of class $\psi$ are packed in such bins. We refer to both the solution and the number of those bins in this solution by $\opt_1$.  Similarly, for $I_2$ we consider a solution that maximizes the number of bins with items of total size in $[1,1+3\delta)$ subject to the constraint that for every $\psi\in \Psi$ the number of items of class $\psi$ that are packed in such bins is at least $\frac{u_{\psi}}{2}$.  We refer to both the solution and the number of those bins in this solution by $\opt_2$.  Then, we have the following observation by the fact that $\optn$ is an optimal nice solution.
\begin{observation}
The number of near exact covered bins in $\optn$ equals $\opt_1+\opt_2$.
\end{observation}

Thus, by definition the following holds.  Assume that we can find such $\apx_1,\apx_2$ in polynomial time, such that $(1-\kappa\eps)\opt_1+C_1\leq \apx_1$ and $(1-\kappa\eps)\opt_2+C_2\leq \apx_2$ for a given constant  $\kappa$ that is not a function of $\eps$ or $\delta$, and for given constants $C_1$ and $C_2$ (that may depend on $\eps,\delta$). Then we obtain an AFPTAS for $\pr$.
This scheme has an asymptotic approximation ratio of $\frac{1}{(1-\kappa\eps)(1-2\eps)}$.

\section{Approximating $I_1$\label{I1sec}}
Here, we apply a method similar to the AFPTAS's for the bin packing problem \cite{FerLue81,KarKar82}.  These methods are based on linear grouping of the items, formulating a set of so-called {\em bin configurations}, solving a linear program for deciding how many bins are packed with each configuration, and then round up the resulting fractional solution.  Most applications of this procedure invoke an approximated separation oracle to the dual linear program and use this approximated separation oracle to approximately solving the (primal) linear program.  Here we simplify this step using the assumption that every item has size at least $\delta \geq \frac{\Delta}{8}$ where $\Delta>0$ is a fixed constant.  Using this simplification, we solve exactly the primal linear program without using this mechanism of \cite{KarKar82}.  When we consider $I_1$ with respect to $\opt_1$ we say that a bin is a {\em good bin} if the total size of its items is in $[1+\delta,1+\Delta)$.

\subsection{Linear grouping of each class of huge items}  

For every $\psi\in \Psi$, we apply linear grouping of $\h_{\psi}$ into $\frac{1}{\eps^2 \delta}$ subsets.  That is, for every $\psi$, if there are strictly less than $\frac{1}{\eps^2 \delta}$ items of class $\psi$ in the input $I_1$, then every item of the class has its own set in the collection of sets $\h_{\psi}^{\alpha}$ for $\alpha=1,2,\ldots ,\frac{1}{\eps^2 \delta}$, and no rounding is applied so the rounded up size of an item in the class equals its size.  Furthermore, we assume that $\h_{\psi}^1=\emptyset$, and for $\alpha\geq 2$, the set $\h_{\psi}^{\alpha}$ contains the $\alpha-1$-th largest index item of the class, and if there is no such item, then the last set is the empty set.  

If for the given value of $\psi$ we have at least  $\frac{1}{\eps^2 \delta}$ items of class $\psi$ in the input $I_1$, then we require the following.  First,  $|\h_{\psi}^1| \geq |\h_{\psi}^2| \geq \cdots \geq |\h_{\psi}^{1/(\eps^2\delta)}| \geq |\h_{\psi}^1|-1$. Second, the indexes of the items in $ \h_{\psi}^1$ are the largest ones in the class, the indexes of the items in $ \h_{\psi}^2$ are the next largest ones in the class, and so on.  In this case, we apply rounding and we let the rounded up size of an item of the class to be the largest size of an item in its subset.  

Observe that for every $\psi$ we have $|\h_{\psi}^1| \leq 2\eps^2\delta \cdot |\h_{\psi}|$.  We denote by $s'_i$ the rounded up size of item $i$.  Let $\opt'_1$ be a solution for the input items in $I_1$ where the size of each item is its rounded up size such that $\opt'_1$ maximizes the number of good bins. The use of the linear grouping to approximate $I_1$ is motivated by the following lemma.  In what follows we will approximate $\opt'_1$.

\begin{lemma}\label{lingroup1}
We have $(1-\eps) \cdot \opt_1 \leq \opt'_1$.  Furthermore, let $\sol$ be a feasible solution to the rounded instance whose number of good bins is also denoted as $\sol$.  Then, packing each item of the original (non-rounded) instance exactly as in $\sol$, results in a solution whose number of near exact covered bins is at least $\sol$. 
\end{lemma}
\begin{proof}
First, consider the second part of this lemma.  This claim holds as every good bin $\B$ of $\sol$ (in the rounded instance) has at most $\frac{2}{\delta}$ (huge) items of $I_1$. Thus, when we compute the difference in the total size of items in $\B$ with respect to the rounded up size and with respect to the size, then this difference is non-negative and at most $\frac{2}{\delta}\cdot \delta^3$.  Therefore, if $\B$ was one of the good bins in $\sol$ of total (rounded up) size in $[1+\delta,1+\Delta)$, then the total size of its items is in $[1+\delta - 2\delta^2,1+\Delta)$.  Thus, in this case $\B$ is a near exact covered bin in this solution of the original instance.  This suffices for this claim.

It remains to prove the inequality $(1-\eps) \cdot \opt_1 \leq \opt'_1$.  We consider the solution $\opt_1$ and construct a feasible solution for the rounded instance based on $\opt_1$.  For a bin $\B$ whose total (original) size of items  is either at most $1+3\delta$ or at least $1+\Delta$, we leave the items in $\B$ as they were in $\opt_1$.  Next, we claim that the number of those bins with total size in $[1+\delta,1+\Delta)$ is at least their number in $\opt_1$.  This is so because by rounding up the size of each item we can only increase the total size of items, and since we round the size of an item to a size of another item in the same class, such an increase is not larger than $\delta^3$.  Using the fact that $\B$ has at most $2/\delta$ items this last claim holds.  Consider the other bins, and we refer to such bin as {\em free bin} whose items are called {\em free items}.

Next, for the purpose of this proof we re-index the items of every class so that the free items appear first (sorted by the original index of the items in the class) and only afterwards the non-free items.  We repack the free bins as follows.  For a class $\psi$ that has less than $\frac{1}{\eps^2 \delta}$ items we do not change the assignment of the items of the class (there is no rounding of such class) and we say that every item of the class is replaced by itself.  For other classes, we apply the following process.
For a free item of class $\h_{\psi}$ whose index (in the class) is $i$, we pack it in the bin of $\opt_1$  that used to have the item of index $i-\lceil 2\eps^2\delta v_{\psi}\rceil$ in this class. We say that item $i-\lceil 2\eps^2\delta v_{\psi}\rceil$ is replaced by $i$. Note that the rounded size of $i-\lceil 2\eps^2\delta v_{\psi}\rceil$ is at most the (original) size of $i$.
If $i-\lceil 2\eps^2\delta v_{\psi}\rceil\leq 0$ (meaning there is no such item), then we pack $i$ in a new dedicated bin.   We have that if a bin $\B$ in $\opt_1$ has only items that are replaced by items, then the total rounded up size of the replacing items of $\B$ (in the new solution) is at most the total original size of items in $\B$ and this is at most $1+\Delta$.  On the other hand, every such replaced item has a size that is at most $\delta^3$ smaller than the item it replaces. Since $\B$ has less than $\frac 2{\delta}$ items, the total (rounded) size of the items of the bin after this replacement is not smaller than $1+3\delta-\frac{2}{\delta}\cdot \delta^3 >1+\delta$.

It remains to show that the number of free bins containing items that are not replaced by other items is at most $\eps\opt_1$.  First note that every good bin of $\opt_1$ has at most $\frac{2}{\delta}$ items out of at least $\sum_{\psi\in \Psi} \frac{v_{\psi}}{2}$ items that participate in such bins.  Therefore, $$\opt_1\geq \frac{\delta}{4} \sum_{\psi\in \Psi} v_{\psi} \ .$$  Second, for a class $\psi$, the number of items that are not replaced is at most $\lceil 2\eps^2\delta v_{\psi}\rceil\leq 2\eps^2\delta v_{\psi}+1$. But such items exist only if $\eps^2\delta v_{\psi} \geq 1$ so we have at most $3\eps^2\delta v_{\psi}$ such items of class $\psi$.
Thus, by summing over all classes, we have at most $\sum_{\psi\in \Psi} 3\eps^2\delta v_{\psi}$ such items. This is also a valid upper bound on the number of free bins containing an item that is not replaced.  Finally, we have $$\sum_{\psi\in \Psi} 3\eps^2\delta v_{\psi}= 3\eps^2\delta \cdot \sum_{\psi\in \Psi} v_{\psi} \leq 12\eps^2 \opt_1\leq \eps\opt_1$$ where the last inequality holds using $\eps\leq \frac{1}{12}$.
\qed\end{proof}

\subsection{The configuration LP}  

We define a bin configuration to describe a packing of one bin.  These bin configurations are used next to formulate the so-called configuration integer program that directs our algorithm.  We would like to approximate $\opt'_1$ and consider the instance with item set $I_1$ where the size of every item $i$ is the rounded up size of the item $s'_i$, i.e., the {\em rounded-up instance}.  Our goal is to maximize the number of good bins. 

Formally, a {\em bin configuration} is a multi-set of sizes of items in the rounded up instance where the total size of items in this multi-set is strictly less than $1+\Delta$.  Since the size of every item in this instance is at least $\delta$, each multi-set of items described by a configuration has at most $\frac{2}{\delta}$ items.  Furthermore, we say that a bin configuration has a unit reward if the total size of its items is in $[1+\delta,1+\Delta)$ and otherwise it has zero reward.  Next, we prove that there is polynomially many bin configurations (for a fixed value of $\delta$).

\begin{lemma}
There are at most $(\frac{1}{\eps^2\delta^4}+1)^{2/\delta}$ bin configurations.
\end{lemma}
\begin{proof}
In the rounded-up instance the number of different sizes is at most $\frac{1}{\eps^2\delta^4}$.  We consider the items in a bin configuration as a sequence containing $\frac{2}{\delta}$ positions.  In such sequence the $i$-th position is the index (in the list of sizes) of the size of the $i$-th item in the sequence where if no such item exists, then the corresponding position is $0$.  In total there are exactly $\frac{2}{\delta}$ positions and each of which is described as a non-negative integer that is at most  $\frac{1}{\eps^2\delta^4}$.  So the claim holds.
\qed\end{proof}

Observe that this last bound on the number of bin configurations is polynomial (of a constant degree) in $\frac{1}{\eps}$ and independent of the input encoding length.  So in order to design an AFPTAS we can have a step whose time complexity is polynomial in the number of bin configurations.  Next, we formulate the configuration integer program. 

For this formulation we treat a bin configuration as a vector of non-negative integers where the $i$-th component of the vector is the number of items of the $i$-th size in the multi-set described by the configuration.  We let $\C$ denote the set of all bin configurations.
 The decision variable $x_c$ for a configuration $c\in \C$ is the number of bins with configuration $c$.  For the $i$-th size in the input we denote by $c_i$ the $i$-th component of configuration $c\in \C$ and by $\nu_i$ the number of items in the rounded up instance of this size.  We denote the subset of $\C$ consisting of all configurations with unit reward by $\C_1$.  We denote by $\sigma$ the index set of sizes of items in the rounded up instance. The {\em configuration integer program} is the following integer program.
 
 \begin{eqnarray*}
 \max& \sum_{c\in \C_1} x_c & \\
 s.t. & \sum_{c\in \C} c_i \cdot x_c = \nu_i &\forall i \in \sigma \\
 & x_c \geq 0& \forall c\in \C \ .
 \end{eqnarray*}
 The number of decision variables is the number of bin configurations that is at most $(\frac{1}{\eps^2\delta^4}+1)^{2/\delta}$. The number of constraints (excluding the non-negativity constraints) is  $|\sigma|\leq \frac{1}{\eps^2\delta^4}$.  The {\em configuration LP} is the linear programming relaxation obtained from the above integer program by allowing the variables to be non-integers.
 
Our algorithm formulates the configuration LP and solves it (optimally) using the ellipsoid algorithm or another polynomial time algorithm for solving linear programs.  This step runs in polynomial time as the dimension (i.e., the number of decision variables) and the number of constraints are upper bounded by a polynomial in $\frac{1}{\eps}$.  The maximum encoding length of a number that appears in  this linear program is $O(\log n)$. Therefore, the polynomial time algorithms for solving a linear program runs in time that is upper bounded by a polynomial in $\frac{1}{\eps}$ times a polylog in the number of items $n$.  Furthermore, we can assume that this algorithm outputs a basic optimal solution (by applying a basis-crashing algorithm like \cite{beling98}). We denote an optimal basic solution of this linear program by $x^*$.  We have that $\opt'_1 \leq \sum_{c\in \C_1} x^*_c$ as we show next using the fact that $\opt'_1$ defines an integer feasible solution to the configuration LP.

\begin{lemma}\label{lemopt'}
We have $\opt'_1 \leq \sum_{c\in \C_1} x^*_c$.
\end{lemma}
\begin{proof}
Based on $\opt'_1$ we define a bin configuration for every bin of this solution, that is, the multi-set of sizes packed into this bin.  Observe that without loss of generality every bin in this solution has items of total size less than $1+\Delta$, so indeed there is a configuration in $\C$ with this multi-set of sizes.  Then, we set the integer feasible solution $x^o$ as follows.  For every $c\in \C$, the value of $x^o_c$ is the number of bins in the solution $\opt'_1$ with configuration $c$.  Since every item is packed into exactly one bin, the constraint $\sum_{c\in \C} c_i \cdot x_c = \nu_i$ is satisfied for all $i\in \sigma$. Thus, indeed $x^o$ is a feasible integer solution for the configuration LP.  Its objective function value as a solution to this linear program is exactly $\opt'_1$. The claim holds as $x^*$ is an optimal solution for the linear program.
\qed\end{proof}
 
We let $x'_c=\lfloor x^*_c \rfloor$ for all $c\in \C$. We pack a subset of items based on this integer vector $x'$ using the following process.  For every $c\in \C$ we have $x'_c$ bins packed according to $c$.  For each such bin $\B$ packed according to $c$ and every size $i\in \sigma$, we pick $c_i$ items of the $i$-th size and pack them into $\B$.  These picked items are not picked to other bins.  We have enough items of each size as $c$ is non-negative for all $c\in \C$ and $x'_c\leq x^*_c$ so $$\sum_{c\in \C} c_i \cdot x'_c \leq  \sum_{c\in \C} c_i \cdot x^*_c =\nu_i$$ holds for all $i\in \sigma$.  Additional items that were not packed by this process are packed into dedicated bins and do not contribute to the objective function value of $x'$.  In this way we output a feasible solution $\apx_1$ satisfying the required performance guarantee as we establish next.
\begin{lemma}
We have $$\apx_1 \geq (1-\eps) \cdot \opt_1  -  \frac{1}{\eps^2\delta^4}.$$
\end{lemma} 
\begin{proof} We have the following.
\begin{eqnarray*}
\apx_1 & =& \sum_{c\in \C_1} x'_c \\
& \geq & \sum_{c\in \C_1} x^*_c - \frac{1}{\eps^2\delta^4} \\
& \geq & \opt'_1 -  \frac{1}{\eps^2\delta^4} \\
& \geq & (1-\eps) \cdot \opt_1  -  \frac{1}{\eps^2\delta^4},
\end{eqnarray*}
 where the first inequality holds as $x^*$ is a basic optimal solution for the linear program, the second inequality by Lemma \ref{lemopt'}, and the last inequality by Lemma \ref{lingroup1}.
\qed\end{proof}

\section{Approximating $I_2$\label{I2sec}}
Here, we will say that a bin $\B$ is a {\em good bin} if the total size of items in $\B$ is in the interval $[1,1+3\delta)$. Recall that $\opt_2$ is a solution maximizing the number of good bins among all solutions that for every $\psi$ have at least $u_{\psi}/2$ items packed in good bins.  We have that without loss of generality every good bin in $\opt_2$ that has a non-huge item have items of total size in $[1,1+\delta)$, and there is at most one bin that is not a good bin containing non-huge items.  To see the first observation we can repack one non-huge item placed in a good bin with items of total size larger than $1+\delta$ into a dedicated bin, and repeat as long as this first property does not hold.  To verify the second observation, move all non-huge items packed into bins that are not good bins, into one bin.  This process cannot decrease the number of good bins and satisfy the second property.  If the first property stops to hold, we partition this one new bin into several ones, where at most one of these new bins is not a good bin.

\paragraph{Additional guessing step.} Our first step is to guess the value of $\beta$ defined as the number of good bins in $\opt_2$ containing non-huge items.  Since $\beta$ is a non-negative integer not larger than the number of non-huge items and in particular $\beta\leq n$, we can enumerate all possibilities for the value of $\beta$. For each such possibility, we construct a feasible solution for our problem instance (see below). Last we choose the best feasible solution (among all iterations of the exhaustive enumeration implementing this guessing step).  In what follows, $\beta$ is the value of the guessed information.

\paragraph{Another classification of items.} We let $X$ denote the set of non-huge items.  We sort the items in $X$ in a non-increasing order of size (breaking ties based on the index of the item).  We classify the items in $X$ as follows.  The first $\min\{ |X| , \lceil \frac{\beta+1}{\eps} \rceil \}$ items in the sorted list of $X$ are {\em large items}, and we let $L=\h_0$ be the set of large items.  As will be clear later on, we also let $L$ be the class $0$ of huge items (although they are not huge items).  The next $\beta$ items in the sorted list are {\em medium items} and their set is denoted as $M$.   Last, the remaining items are {\em small items} whose set $S$ is $S=X\setminus (L \cup M)$.  If $|X|-|L| < \beta$, then $M=X\setminus L$ and $S=\emptyset$.  Furthermore, we let $\Psi'=\Psi\cup \{ 0 \}$ be the set of classes of huge items including the class $0$ containing the large items, and we let $\h'=\h\cup L$.

\subsection{Linear grouping of each class in $\h'$ and excluding the items in $M$}
We apply a linear grouping step similar to \cite{CJK01,JS03}, that is, this time the rounded size will be rounded down value and not rounded up as we did when approximating $I_1$.
For every $\psi\in \Psi'$ we apply linear grouping of $\h_{\psi}$ into $\frac{1}{\eps^2 \delta}$ subsets.  That is, for every $\psi$, if there are strictly less than $\frac{1}{\eps^2 \delta}$ items of class $\psi$ in the input $I_2$, then every item of the class has its own set in the collection of sets $\h_{\psi}^{\alpha}$ for $\alpha=1,2,\ldots ,\frac{1}{\eps^2 \delta}$. In this case, no rounding is applied so the rounded down size of an item in the class equals its size and we assume that $\h_{\psi}^1=\emptyset$.  

If for the given value of $\psi\in \Psi'$ we have at least  $\frac{1}{\eps^2 \delta}$ items of class $\psi$ in the input $I_2$, then we will require that $|\h_{\psi}^1| \geq |\h_{\psi}^2| \geq \cdots \geq |\h_{\psi}^{1/(\eps^2\delta)}| \geq |\h_{\psi}^1|-1$ and that the indexes of the items in $ \h_{\psi}^1$ are the largest ones in the class, the indexes of the items in $ \h_{\psi}^2$ are the next largest ones in the class, and so on.  In this case, we apply rounding and we let the rounded down size of an item of the class to be the smallest size of an item in its subset.  Observe that for every $\psi$,  we have $|\h_{\psi}^1| \leq 2\eps^2\delta \cdot |\h_{\psi}|$.  

We denote by $s'_i$ the rounded down size of item $i$ and for item in $S\cup M$ we let the rounded down size of the item be its size.   For a subset of items $\Lambda$ we let $s(\Lambda)=\sum_{i\in \Lambda} s'_i$ be its total rounded size.  Furthermore, we allow the algorithm to pack temporarily the items in $S$ fractionally. That is, we treat the small items as  {\em fluid} or {\em sand} of total size $s(S)$, and every bin may pack an arbitrary sized sand as long as the total size of the packed sand is not larger than $s(S)$.  We define the {\em rounded down instance of the problem} as the input with item set $I_2\setminus M$, the size of every item is the rounded down size of the item, and sand of total size $s(S)$ that can be packed fractionally.

 Let $\opt'_2$ be a solution for the rounded down instance such that $\opt'_2$ maximizes the number of good bins subject to the constraint that the number of good bins with non-zero sand is at most $\beta$. 
The use of the linear grouping to approximate $I_2$ is justified by the following two lemmas. 

\begin{lemma}\label{lingroup2a}
There is a polynomial time algorithm accomplishing the following task.  The input consists of a solution $\sol'$ for the rounded-down instance whose number of good bins is $\sol'$ where we assume that the number of good bins with non-zero space for sand in $\sol'$ is at most $\beta$. The output is a solution to the original instance $I_2$ whose number of near exact covered bins is at least $\sol'$. 
\end{lemma}
\begin{proof}
First, we consider the packing of huge or large items in $\sol'$, that is of the original sized items and without modifying the allocated sand for each bin. We identify the set $\zeta$ of good bins in $\sol'$ containing non-zero space for sand (as a solution to the rounded down instance).  
Consider a bin $\B$ in the resulting solution, and we show that if $\B$ was a good bin in $\sol'$, then either $\B$ is a near exact covered bin in the resulting solution, or we can can repack some of the large items used to be packed into $\B$ (and move these repacked items into dedicated bins), so that the resulting bin of the items and sand left in $\B$ is near exact covered. Thus, the upper bound of $\beta$ on the number of good bins containing sand will continue to hold.

Since the rounding of items was rounding down and $\B$ was a good bin in the rounded down instance, the total (original) size of its items and sand,  whose size is not modified, is at least $1$.  If $\B$ has no large items, then it has at most $\frac{2}{\delta}$ items plus some additional sand, and the maximum difference between the original size of an item and its rounded-down size is at most $\delta^3$.  Since $\B$ was a good bin, its total rounded-down size is smaller than $1+3\delta$.  Using the fact that the total size is at most $1+3\delta+\frac{2}{\delta}\cdot \delta^3 < 1+\Delta$, the claim follows.
Next, consider the case where some items in $\B$ are large.  If the total size of the items in $B$ is less than $1+\Delta$, then we are done.  Otherwise, we start deleting from $\B$ the large items, one after the other.  The process stops either once there are no large items, or when the total size of the items and sand left in $\B$ is  less than $1+\Delta$.  By the proof of the case where the original bin had no large items, we conclude that when the process ends it must be the case that the total size of the items and sand left in $\B$ is less than $1+\Delta$.  

At this point, the set $\zeta$ contains all good bins with non-zero sand and by the assumption of the lemma, $\zeta$ has at most $\beta$ bins.  Our next goal is to replace the sand in the bins of $\zeta$ by the items in $S\cup M$, so that the number of near exact covered bins will not decrease.  First, note that we can decrease the amount of packed sand in some good bins to ensure that if a good bin has non-zero sand, then its total size of items and sand is exactly $1$.  With a slight abuse of notation, we let $\zeta$ be the resulting set of good bins with non-zero sand. This is a subset of the original set $\zeta$ so it has at most $\beta$ bins.

We process the bins in $\zeta$, one by one.  Consider the current bin $\B\in \zeta$.  We remove the sand from $\B$ and we start adding items from $S$ until the first time in which the added item is about to increase the total size of items in $\B$ to be at least $1$.  This last item from $S$ is not added to $\B$ and instead we add one item from $M$.  Then, we conclude that the resulting bin has size not smaller than $1$ and not larger than $1+\delta$ where the upper bound follows as medium items are not huge.  The items that we added to $\B$ are deleted from the corresponding sets ($S$ or $M$) and we continue to process the next bin from $\zeta$.  

Observe that the total size of small items that we pack to a good bin $\B$ is not larger than the size of sand packed into $\B$ in the solution $\sol'$. Therefore, by summing over all bins in $\zeta$, we have sufficiently many small items to pack into all bins in $\zeta$.  We have enough medium items as every bin in $\zeta$ gets one medium item and $|\zeta|\leq \beta \leq |M|$ where the second inequality holds in cases there is sand in the instance (and otherwise $\zeta=\emptyset$ and this step does not exist).  If the process leaves some unpacked small or medium items (after processing all bins in $\zeta$), then we pack each such item in its own dedicated bin without modifying the number of near exact covered bins.
\qed\end{proof}

Next, we consider the other direction showing that approximating the rounded down instance is sufficiently close to approximate the original instance. 

\begin{lemma}\label{lingroup2}
We have $(1-3\eps) \cdot \opt_2 -1 \leq \opt'_2$. 
\end{lemma}
\begin{proof}
We consider the solution $\opt_2$ and modify it in two steps.  First, we consider the rounded-down instance including the items in $M$ (the rounded down size of such item is its original size) and construct a feasible solution $\sol$ whose number of good bins is at least $(1-2\eps)\opt_2-1$. Then, in the second step we construct another solution for the rounded down instance without medium items whose number of near exact covered bins is at least $(1-3\eps)\opt_2-1$.   In both steps the allocation of small items is not modified and we will not use the ability to pack them fractionally when we consider the rounded down instance.  Furthermore, we will say below that every small item $i$ is replaced by itself.

Consider the first step. Assume without loss of generality that in $\opt_2$ every bin of total size at least $1+\delta$ satisfies that the bin consists only of huge items (so every item has size at least $\delta$).  For a bin $\B$ whose total size of items in the original instance is either less than  $1$ or at least $1+\delta$, we leave the items in $\B$ as they were in $\opt_2$.  Next, we claim that the number of those good bins is the same as it were in $\opt_2$.  To show this first note that by rounding down the size of each item, we can only decrease the total size of items, and since we round the size of an item to a size of another item in the same class, such decrease is at most $\delta^3$.  Since by our assumption a good bin $\B$ has no large items and the size of small items is not modified, $\B$ has at most $2/\delta$ huge items, and so this last claim holds.  Consider the other bins, and we refer to such bin as {\em free bin} whose items are called {\em free items}.

Next, for the purpose of this proof we re-index the items of every class of $\h'$ so that the free items appear first (sorted by the original index of the items in the class) and only afterwards the non-free items.  We repack the free bins as follows.  For a class $\psi\in \Psi'$  that has less than $\frac{1}{\eps^2 \delta}$ items, we do not change the assignment of the items of the class (there is no rounding of such class) and we say that every item of the class is replaced by itself.  

For other classes, we apply the following process where we let $u_0=|L|$ be the number of large items.
For a free item of class $\h_{\psi}$ whose index (in the class) is $i$, we pack it in the bin of $\opt_2$  that used to have the item of index $i-\lceil 2\eps^2\delta u_{\psi}\rceil$, and we say that item $i-\lceil 2\eps^2\delta u_{\psi}\rceil$ is replaced by $i$. We note that the rounded size of $i-\lceil 2\eps^2\delta u_{\psi}\rceil$ is at least the (original) size of $i$.
If $i-\lceil 2\eps^2\delta u_{\psi}\rceil\leq 0$ (meaning there is no such item), then we pack $i$ in a new dedicated bin.   

We have that if a bin $\B$ in $\opt_2$ has only items that are replaced by (perhaps other) items, then the total rounded down size of the replacing items of $\B$ (in the new solution) is at least the total original size of items in $\B$ and this is at least $1$.  On the other hand, every such replaced item has a size that is at most $\delta^3$ larger than the item it replaces. Since the free bin $\B$ has less than $\frac 2{\delta}$ huge items, the total (rounded) size of the bin after this replacement is not larger than $1+\delta+\frac{2}{\delta}\cdot \delta^3 <1+3\delta$ if there are no large items in $\B$, so this is a good bin.  

Note that if there are also large items in $\B$, then the last proof shows that the total rounded size of the replaced items of the huge and small items is less than $1+3\delta$, and we can add one large item after the other until the first iteration in which the resulting bin is a good bin.  Since the size of every large item is at most $\delta$, the resulting bin will be indeed a good bin.

It remains to show that the number of free bins containing items that are not replaced by other items is at most $2\eps\opt_2+1$.  First note that every good bin of $\opt_2$ has at most $\frac{2}{\delta}$ huge items and the number of huge items participating in good bins is at least $\sum_{\psi\in \Psi} \frac{u_{\psi}}{2}$.  Therefore, $$\opt_2\geq \frac{\delta}{4} \sum_{\psi\in \Psi} u_{\psi} \ .$$  Second, for a class $\psi\in \Psi$, the number of items that are not replaced is at most $\lceil 2\eps^2\delta u_{\psi}\rceil\leq 2\eps^2\delta u_{\psi}+1$, but such items exist only if $\eps^2\delta u_{\psi} \geq 1$ so we have at most $3\eps^2\delta u_{\psi}$ such items of class $\psi\in \Psi$.
Summing over all such classes, we have at most $$\sum_{\psi\in \Psi} 3\eps^2\delta u_{\psi}$$ such items and this is also a valid upper bound on the number of free bins containing a huge item that is not replaced.
Note that $$\sum_{\psi\in \Psi} 3\eps^2\delta u_{\psi}= 3\eps^2\delta \cdot \sum_{\psi\in \Psi} u_{\psi} \leq 12\eps^2 \opt_2\leq \eps\opt_2$$ where the last inequality holds using $\eps\leq \frac{1}{12}$.
Next consider the number of large items that are not replaced.  These are at most $3\eps^2\delta |L|$ items.  By definition of $L$, we know that $|L|\leq \frac{\beta+1}{\eps} +1$, and $\beta \leq \opt_2$.  So the number of non-replaced large items is at most $$3\eps^2\delta \cdot ( \frac{\beta+1}{\eps} +1) \leq \eps\opt_2 +1$$ using $\delta\leq \frac 14$. This is a valid upper bound on the number of free bins containing large items that are not replaced.  This conclude the first step where we have considered the rounded-down instance including the items in $M$ and construct a feasible solution denoted as $\sol$ whose number of good bins is at least $(1-2\eps)\opt_2-1$.

In the last step of the proof, we modify $\sol$ into a feasible solution that does not pack the items in $M$.  
We replace the medium items in $\sol$ by large items by repacking the items of at most $\eps\beta$ good bins and  other bins containing large items.  In this repacking process every bin containing $i$ medium items will be assigned $i$ large items that were not packed there prior to this step.  This repacking increases the total rounded size of items in the bin so it will be at least $1$. However, if it becomes at least $1+3\delta$, we remove one large item after another until the total rounded size becomes less than $1+3\delta$.  Thus this repacking will generate the required solution for the rounded instance.  

However, we need to show that there is a set of at most $\eps\beta$ good bins such that together with all non-good bins contains at least $|M|$ large items.  This is so, as in $\opt_2$ there is a set of $\beta+1$ bins containing all large items and $|L|\geq \frac{\beta+1}{\eps}$ or $M=\emptyset$.  By the pigeonhole principle, out of the bins of $\opt_2$ with large items, there is a set of $\eps\beta$ good bins such that together with the unique non-good bin having large items we get a subset $BINS$ of the bins with at least $\beta$ large items.  So the claim follows.  
\qed\end{proof}

\subsection{The configuration LP}  

We define a bin configuration to describe a packing of one bin.  These bin configurations will be used to formulate the configuration integer program whose linear programming relaxation is used to direct our algorithm.  Both the definition of bin configurations as well as the configuration integer program differ from the ones we have considered in Section \ref{I1sec} to approximate $I_1$. We would like to approximate $\opt'_2$ and consider the rounded down instance.  

Formally, a bin configuration is a multi-set of sizes of items in $\h'$ together with a non-negative allocated space for sand such that the total rounded down size of the items together with the space for sand is at most $1+3\delta$.   It is clear that without loss of generality we conclude that for every bin configuration corresponding to a good bin either the space for sand is zero, or it is exactly $1$ minus the total size of the items in $\h'$ in this bin configuration.  However, for a multi-set of items of total size smaller than $1$ there are two configurations, one with sand so that it is a good bin and one without sand (and in this case it is not a good bin). Therefore, the information regarding the allocated space for sand in the configuration is implied by the multi-set of items in $\h'$ together with one additional bit of information. Furthermore, we say that a bin configuration has a unit reward if the total size of its items and sand is in $[1,1+3\delta)$ and otherwise it has zero reward.  However, unlike the case of $I_1$, the number of bin configurations is not polynomially bounded and we prove the following upper bound.

\begin{lemma}
There are at most $2\cdot (n+1)^{\frac{1}{\eps^2\delta^4}}$ bin configurations.
\end{lemma}
\begin{proof}
The items in $\h'$ have at most $\frac{1}{\eps^2\delta^4}$ different sizes of items, and for each such size the number of items in the configuration is a non-negative integer that is at most $n$.  Every such multi-set of items has at most two configurations where one of those has unit reward and at most one of those has zero reward.  Therefore, a valid upper bound on the number of configurations is $2(n+1)^{\frac{1}{\eps^2\delta^4}}$.
\qed\end{proof}

Observe that this last bound on the number of bin configurations is not polynomial (of a constant degree).  So in order to design an AFPTAS we cannot enumerate all bin configurations. 
Next, we formulate the configuration integer program. 
Here, we consider bin configuration as a vector of non-negative integers where the $i$-th component of the vector is the number of items of the $i$-th size in the multi-set described by the configuration.  We let $\C$ denote the set of all bin configurations.  For $c\in \C$, we let $s(c)$ denote the space for sand in the configuration $c$.   We denote the subset of $\C$ consisting of all configurations with unit reward by $\C_1$.  We let $\sigma$ denote the index set of sizes of items in $\h'$.  For $i\in \sigma$,  let $\nu_i$ denotes the number of items of (rounded down) size $s^i$ in the rounded instance.

The decision variable $x_c$ for a configuration $c\in \C$ is the number of bins with configuration $c$.  For  $i\in \sigma$, we denote by $c_i$ the $i$-th component of configuration $c\in \C$.   The configuration integer program is the integer program stated below.  Here, we use the assumption that if a solution to this integer program does not use all items of a given size (or the full amount of sand in the instance), then additional items can be packed into dedicated bins without changing the objective function value.  This motivates our use of inequalities (instead of equalities) in the constraints (beside the non-negativity constraints) of the configuration integer program.
 
 \begin{eqnarray*}
 \max& \sum_{c\in \C_1} x_c & \\
 s.t. & \sum_{c\in \C} c_i \cdot x_c \leq \nu_i &\forall i \in \sigma \\
 & \sum_{c\in \C_1} s(c) \cdot x_c \leq s(S)& \\
 & \sum_{c\in \C_1: s(c)>0} x_c \leq \beta& \\ 
 & x_c \geq 0& \forall c\in \C \ .
 \end{eqnarray*}
 The number of decision variables is the number of bin configurations. The number of constraints (excluding the non-negativity constraints) is  $|\sigma|+2 \leq \frac{1}{\eps^2\delta^4}+2$.  The {\em configuration LP} is the linear programming relaxation obtained from the above integer program by allowing the variables to be non-integers.

Our algorithm considers the configuration LP without stating its formulation, and solve it approximately within a multiplicative factor of $1-\eps$ using the ellipsoid algorithm via the column-generation approach of \cite{KarKar82}.  We postpone the details of this step and assume that we are given a feasible solution $x^*$ to the configuration LP such that the support of $x^*$ has at most polynomial number of elements, and satisfying that for every other feasible solution $\hat{x}$ for the configuration LP, we have $\sum_{c\in \C_1} x^*_c \geq (1-\eps) \cdot \sum_{c\in \C_1} \hat{x}_c$. Furthermore, using a basis-crashing algorithm \cite{beling98} we can assume that the support of $x^*$ has at most $\frac{1}{\eps^2\delta^4}+2$ elements. We have that $(1-\eps) \cdot \opt'_2 \leq \sum_{c\in \C_1} x^*_c$ as we show next using the fact that $\opt'_2$ defines an integer feasible solution to the configuration LP.

\begin{lemma}\label{lemopt2}
We have $(1-\eps) \cdot \opt'_2 \leq \sum_{c\in \C_1} x^*_c$.
\end{lemma}
\begin{proof}
Based on $\opt'_2$ we define a bin configuration for every bin of this solution, that is, the multi-set of sizes of $\h'$ packed into this bin.  Observe that without loss of generality every bin in this solution has items of total size at most $1+3\delta$. So indeed there is a configuration in $\C$ with this multi-set of sizes where we may have identical configurations where one copy belongs to $\C_1$ with non-zero sand packed into it and the other copy to $\C\setminus \C_1$. We pick the first such copy if the bin is a good bin in $\opt'_2$ and the second copy if it is not a good bin.  Then, we let the integer feasible solution $x^o$ be defined as follows.  For every $c\in \C$, the value of $x^o_c$ is the number of bins in the solution $\opt'_2$ with configuration $c$.  Since every item is packed into exactly one bin, the constraint $\sum_{c\in \C} c_i \cdot x_c \leq \nu_i$ is satisfied for all $i\in \sigma$.  Similarly the sand resulting from small items is packed once, so the constraint $\sum_{c\in \C_1} s(c) \cdot x_c \leq s(S)$ is also satisfied.  Last, the number of unit reward bins with some small items in $\opt'_2$ is at most $\beta$ by the guessing step. So indeed this is a feasible integer solution for the configuration LP.  Its objective function value as a solution to this linear program is exactly $\opt'_2$. The claim holds as $x^*$ is a $1-\eps$ approximation for the linear program.
\qed\end{proof}
 
We let $x'_c=\lfloor x^*_c \rfloor$ for all $c\in \C$. We pack a subset of items based on this integer solution $x'$ using the following process.  For every $c\in \C$ we have $x'_c$ bins packed according to $c$.  For each such bin $\B$ packed according to $c$ and every size $i \in \sigma$, we pick $c_i$ items of the $i$-th size and pack them into $\B$.  These picked items are not picked to other bins.  We have enough items of each size as $c$ is non-negative for all $c\in \C$ and $x'_c\leq x^*_c$ so $\sum_{c\in \C} c_i \cdot x'_c \leq  \sum_{c\in \C} c_i \cdot x^*_c \leq \nu_i$ (for all $i$).  Additional items of $\h'$ that were not packed by this process are packed into dedicated bins and do not contribute to the objective function value of $x'$. The sand is allocated only to bins where the corresponding configuration belongs to $\C_1$ and such bin gets sand of size that is the minimum amount for which the total size of its items and sand will be at least $1$.  Since $s(c)$ is non-negative for all configurations, we have $\sum_{c\in \C_1} s(c) \cdot x'_c \leq \sum_{c\in \C_1} s(c) \cdot x^*_c \leq s(S)$.  Similarly, since $x'\leq x^*$, we conclude that $ \sum_{c\in \C_1: s(c)>0} x'_c \leq  \sum_{c\in \C_1: s(c)>0} x^*_c\leq \beta$.
  Therefore, $x'$ satisfies the constraints of the configuration LP.  Thus, we can apply Lemma  \ref{lingroup2a}, and obtain a feasible solution to the original instance whose number of near exact covered bins is exactly the objective function value of $x'$ to the configuration LP.  We let  $\apx_2$ be the resulting solution that is a feasible solution to the original instance whose number of near exact covered bins satisfies the following lower bound.
\begin{lemma}
The solution $\apx_2$ has at least $(1-3\eps)\cdot (1-\eps) \cdot \opt_2  -  \frac{1}{\eps^2\delta^4}-3$ near exact covered bins.
\end{lemma}   
\begin{proof} We have the following.
\begin{eqnarray*}
\apx_2 & =& \sum_{c\in \C_1} x'_c \\
& \geq & \sum_{c\in \C_1} x^*_c - \frac{1}{\eps^2\delta^4} -2 \\
& \geq & (1-\eps) \cdot \opt'_2 -  \frac{1}{\eps^2\delta^4} -2 \\
& \geq & (1-3\eps)\cdot (1-\eps) \cdot \opt_2  -  \frac{1}{\eps^2\delta^4}-3,
\end{eqnarray*}
 where the first inequality holds as $x^*$ is a basic solution for the linear program, the second inequality by Lemma \ref{lemopt2}, and the last inequality by Lemma \ref{lingroup2a}. 
 \qed\end{proof}
  It remains to show that we can indeed solve the configuration LP within a multiplicative factor of $1-\eps$ in polynomial time as we prove next.

\subsection{Approximating the configuration LP.}  Here, we refer to the configuration LP as the {\em primal LP}. Since the dimension of this linear program is exponential in $\frac{1}{\eps}$, we consider its dual linear program and refer to it as the dual LP.  This dual LP has number of variables that is polynomial in $\frac{1}{\eps}$ and $\frac{1}{\delta}$, and we plan to use the ellipsoid algorithm to solve this dual LP.  However, the number of constraints of the dual LP is the number of configurations, and if we try to use separation oracle for deciding if the current dual solution is feasible or not, we get an NP-complete problem.  We will use the column-generation method of \cite{KarKar82} to overcome this difficulty.

\paragraph{The dual linear program.} First, we state the dual linear program.  We have a dual variable $y_i$ associated with the primal constraint 
$\sum_{c\in \C} c_i \cdot x_c \leq \nu_i $.  In addition to these dual variables, we let $z_1$ be the dual variable associated with $\sum_{c\in \C_1} s(c) \cdot x_c \leq s(S)$ and $z_2$ be the dual variable associated with $\sum_{c\in \C_1: s(c)>0} x_c \leq \beta$.  The dual LP is the following linear program (once again the algorithm does not list all constraints of this linear program).
\begin{eqnarray}
\min & \sum_{i\in \sigma} \nu_i \cdot y_i + s(S) \cdot z_1 + \beta\cdot z_2& \\
s.t.& \sum_{i\in \sigma} c_i \cdot y_i \geq 0& \forall c\in \C\setminus \C_1 \label{cons1}\\
&  \sum_{i\in \sigma} c_i \cdot y_i + s(c) \cdot z_1 +z_2 \geq 1& \forall c\in \C_1: s(c) >0 \label{cons2}\\
&  \sum_{i\in \sigma} c_i \cdot y_i \geq 1& \forall c\in \C_1: s(c) =0 \label{cons3}\\
& y_i,z_1,z_2 \geq 0& \forall i
\end{eqnarray}

In order to utilize the ellipsoid algorithm for finding a $1+\eps$ approximated solution for the dual LP, we will design approximated separation oracle that given an assignment of values to the decision variables $(\tilde{y},\tilde{z})$ decides in polynomial time if the vector resulting by multiplying every component by $1+\eps$, that is, $(1+\eps) \cdot (\tilde{y},\tilde{z})$ is a feasible solution for the dual LP, or provide a dual constraint that is violated by  $(\tilde{y},\tilde{z})$ namely a configuration $c\in \C$ for which the corresponding constraint is violated.  The time complexity of this separation oracle need to be polynomial in  $n,\frac{1}{\eps}$, and in the binary encoding length of $I_2$.

Given such vector $(\tilde{y},\tilde{z})$, we check in linear time that all its components are non-negative.  If one of these variables is negative, we are done as we found a constraint that is violated by this solution.  Otherwise, we observe that since for every configuration $c\in \C$ the components of $c$ are non-negative, the constraints \eqref{cons1} are satisfied by $(\tilde{y},\tilde{z})$.  In what follows, we denote the $i$-th size in $\sigma$ by $s^i$. 

We define the value of $\tilde{y}_i$ as the value of item of size $s^i$.  A configuration is a multi-set of items of sizes in $\sigma$ whose value is $\sum_{i\in \sigma} \tilde{y}_i \cdot c_i$ and whose size is $\sum_{i\in \sigma} c_i \cdot s^i$.  We treat the values $c_i$ as decision variables of the separation oracle while $\tilde{y}_i, s^i, \tilde{z}_1,\tilde{z}_2$ for all $i$ are constants.  We first round up the values of $\tilde{y}_i$ to the next integer multiple of $\frac{\eps}{n}$, and denote the corresponding rounded value by $y'_i$.  Since every configuration has at most $n$ items in the multi-set, we observe that if $y',\tilde{z}_1,\tilde{z_2}$ satisfies all the constraints of the form constraint \eqref{cons2} and \eqref{cons3}, then $(1+\eps) \cdot \tilde{y}, \tilde{z}_1,\tilde{z_2}$ also satisfies all these constraints. Furthermore, if there exists a constraint in this family that is not satisfied by $(1+\eps) \cdot \tilde{y}, \tilde{z}_1,\tilde{z_2}$,  then $y',\tilde{z}_1,\tilde{z_2}$ violates at least one of these constraints.  

In the following we partition our treatment to configurations $c\in \C_1$ with $s(c)=0$ and the size of $c$ is in the interval $[1,1+\delta)$ (first case of constraint \eqref{cons3}), configurations $c\in \C_1$ with $s(c)=0$ and the size of $c$ is in the interval $[1+\delta,1+3\delta)$ (second case of constraint \eqref{cons3}), and last to configuration $c\in \C_1$ with $s(c)>0$ (constraint \eqref{cons2}).  

\paragraph{The approximated separation oracle for constraint \eqref{cons3} corresponding to configurations with size in $[1,1+\delta)$.} 
Consider the subset of configurations $c\in \C_1$ of items in $\h'$ with size in $[1,1+\delta)$.  The motivation for the following oracle is that in this case we can round up a little bit the sizes and still get a good bin.  For this case, we define the modified size of an item of size $s^i$ as $\hat{s}^i = \lceil \frac{n \cdot s^i}{\delta}\rceil \cdot \frac{\delta}{n}$.  Since every configuration has at most $n$ items, we are guaranteed that if a subset of items has total modified size in the interval $[1,1+\delta)$, then its total size is at most $1+2\delta$, so it is a good bin.   So we consider the following set of integer programs in the vector of decision variables $c$ where the integer program denoted as $IP^1(\iota,\iota')$ is parameterized by two parameters $\iota,\iota'$.

\begin{eqnarray}
\max & \sum_{i\in \sigma} c_i \cdot s^i &  \\
s.t. & \sum_{i\in \sigma} c_i \cdot y'_i = \iota & \label{total-value-cons1} \\
&  \sum_{i\in \sigma} c_i \cdot \hat{s}^i = \iota'& \label{total-mod-size-cons1} \\
& 0 \leq c_i \leq \nu_i & \forall i\in \sigma \label{box-cons1}
\end{eqnarray}
 We solve these problems for every pair of values of $\iota,\iota'$ that are non-negative integer multiples of $\frac{\eps \delta}{n}$ subject to the constraint $\iota,\iota' \leq 2$.
Observe that if we multiply the two constraints \eqref{total-value-cons1} and \eqref{total-mod-size-cons1} by $\frac{n}{\eps\delta}$ we get an equivalent integer program where there are only two constraints (excluding the box constraints of the form \eqref{box-cons1}) and the constraint matrix has only non-negative integer entries  that are at most $\frac{2n}{\eps\delta}$.  Such integer programs can be solved in polynomial time (polynomial in $n,\frac{1}{\delta},\frac{1}{\eps}, |\sigma|$) using e.g. the algorithms of \cite{EW20,JR19}.   Since there are polynomial number of pairs of values $(\iota,\iota')$ satisfying our conditions, we solve all these problems in polynomial time.  We denote by $c^{(\iota,\iota')}$ an optimal solution for the integer program $IP^1(\iota,\iota')$.
 For each pair $(\iota,\iota')$, we check if the configuration $c^{(\iota,\iota')}$ has total size at least $1$ and less than $1+3\delta$, and if so, we check if its corresponding constraint \eqref{cons3} is violated.  
 
 If we found such a violating constraint, we are done.  Otherwise, we argue next that all constraints of the form \eqref{cons3} corresponding to configurations with size in the interval $[1,1+\delta)$ are satisfied.  To see this last claim, assume by contradiction that there is a configuration $c'\in \C_1$ with $s(c')=0$ and total size in the interval $[1,1+\delta)$ whose corresponding constraint is violated.  Let $(\iota,\iota')$ be the pair for which $c'$ is a feasible solution to $IP^1(\iota,\iota')$.  Since $c^{(\iota,\iota')}$ is optimal for this integer program we have that $\sum_{i\in \sigma} c^{(\iota,\iota')}_i s^i \geq \sum_{i\in \sigma} c'_i s^i\geq 1$, and so by our assumption the constraint \eqref{cons3} corresponding to $c^{(\iota,\iota')}$ is not violated.  However, since the constraint corresponding to $c'$ is violated, so $\iota<1$. Therefore, the constraint corresponding to $c^{(\iota,\iota')}$ is also violated.  This contradicts the assumption, so the claim holds.

\paragraph{The approximated separation oracle for constraint \eqref{cons3} corresponding to configurations with size in $[1+\delta,1+3\delta)$.} Consider the subset of configurations $c$ of items in $\h'$ with size in $[1+\delta,1+3\delta)$.  The motivation for the following oracle is that in this case we can round down a little bit the sizes and still get a good bin. Note the difference with the earlier case where decreasing the size of a configuration whose size was slightly larger than $1$ could potentially make it smaller than $1$ so the bin is no longer a good bin.  

For this case, we define the modified size of an item of size $s^i$ as $\hat{s}^i = \lfloor \frac{n \cdot s^i}{\delta}\rfloor \cdot \frac{\delta}{n}$.  Since every configuration has at most $n$ items, we are guaranteed that if a subset of items has total size in the interval $[1+\delta,1+3\delta)$, then its total modified size is in the interval $[1,1+3\delta)$.  So we consider the following set of integer programs denoted as $IP^2(\iota,\iota')$ in the vector of decision variables $c$ where the integer program is parameterized by two parameters $\iota,\iota'$.

\begin{eqnarray}
\min & \sum_{i\in \sigma} c_i \cdot s^i &  \\
s.t. & \sum_{i\in \sigma} c_i \cdot y'_i = \iota & \label{total-value-cons2} \\
&  \sum_{i\in \sigma} c_i \cdot \hat{s}^i = \iota'& \label{total-mod-size-cons2} \\
& 0 \leq c_i \leq \nu_i & \forall i\in \sigma \label{box-cons2}
\end{eqnarray}
We solve these problems for every pair of values of $\iota,\iota'$ that are non-negative integer multiples of $\frac{\eps \delta}{n}$ subject to the constraint $\iota,\iota' \leq 2$.
Observe that if we multiply the two constraints \eqref{total-value-cons2} and \eqref{total-mod-size-cons2} by $\frac{n}{\eps\delta}$, we get an equivalent integer program where there are only two constraints (excluding the box constraints of the form \eqref{box-cons2}) and the constraint matrix has only non-negative integer entries  that are at most $\frac{2n}{\eps\delta}$.  Such integer programs can be solved in polynomial time (polynomial in $n,\frac{1}{\delta},\frac{1}{\eps}, |\sigma|$) using \cite{EW20,JR19}.    Since there are polynomial number of pairs $(\iota,\iota')$ satisfying our conditions, we can indeed solve all these problems in polynomial time.  We denote by $c^{(\iota,\iota')}$ an optimal solution for the integer program $IP^2(\iota,\iota')$.
 For each pair $(\iota,\iota')$ we check if the configuration $c^{(\iota,\iota')}$ has total size at least $1$ and less than $1+3\delta$, and if so, we check if its corresponding constraint \eqref{cons3} is violated.  
 
 If we found such a violating constraint, we are done.  Otherwise, we argue next that all constraints of the form \eqref{cons3} corresponding configurations of size in the interval $[1+\delta,1+3\delta)$ are satisfied.  To see this last claim, assume by contradiction that there is a configuration $c'\in \C_1$ with $s(c')=0$ and total size in the interval $[1+\delta,1+3\delta)$ whose corresponding constraint \eqref{cons3} is violated.  Let $(\iota,\iota')$ be the pair for which $c'$ is a feasible solution to $IP^2(\iota,\iota')$ and there is such a pair of values with $\iota'\geq 1$ since every such configuration of total size in the interval $[1+\delta,1+3\delta)$ has total modified size in the interval $[1,1+3\delta)$.  Since $c^{(\iota,\iota')}$ is optimal for this integer program we have that $\sum_{i\in \sigma} c^{(\iota,\iota')}_i s^i \leq \sum_{i\in \sigma} c'_i s^i < 1+3\delta$. Therefore, as the size of $c^{(\iota,\iota')}_i$ is at least $\iota'\geq 1$, by our assumption, the constraint corresponding to $c^{(\iota,\iota')}$ is not violated so $\iota\geq 1$.  However, since the constraint corresponding to $c'$ is violated, $\iota<1$ and this is a contradiction, so the claim holds.

\paragraph{The approximated separation oracle for constraint \eqref{cons2}.}  Here, we define the modified size of an item of size $s^i$ as rounded down, namely, as $\hat{s}^i = \lfloor \frac{n \cdot s^i}{\delta}\rfloor \cdot \frac{\delta}{n}$.  Since every configuration has at most $n$ items, we are guaranteed that if a subset of items has total modified size of at most $1$, then its total size is at most $1+\delta$.  So we consider the following set of integer programs denoted as $IP^3(\iota,\iota')$ in the vector of decision variables $c$ where the integer program is parameterized by two parameters $\iota,\iota'$.

\begin{eqnarray}
\max & \sum_{i\in \sigma} c_i \cdot s^i &  \\
s.t. & \sum_{i\in \sigma} c_i \cdot y'_i = \iota & \label{total-value-cons} \\
&  \sum_{i\in \sigma} c_i \cdot \hat{s}^i = \iota'& \label{total-mod-size-cons} \\
& 0 \leq c_i \leq \nu_i & \forall i\in \sigma \label{box-cons}
\end{eqnarray}
We solve these problems for every pair of values of $\iota,\iota'$ that are non-negative integer multiples of $\frac{\eps\delta}{n}$ subject to the constraints that  $\iota<1$ and $\iota' \leq 1$. 
Observe that if we multiply the two constraints \eqref{total-value-cons} and \eqref{total-mod-size-cons} by $\frac{n}{\eps\delta}$, we again obtain an equivalent integer program where there are only two constraints (excluding the box constraints of the form \eqref{box-cons}) and the constraint matrix has only non-negative integer entries  that are at most $\frac{n}{\eps\delta}$.  So it can be solved using \cite{EW20,JR19}.   Since there are polynomial number of pairs $(\iota,\iota')$ satisfying our conditions, we can indeed solve all these problems in polynomial time.  We denote by $c^{(\iota,\iota')}$ an optimal solution for the integer program $IP^3(\iota,\iota')$.
 For each pair $(\iota,\iota')$, we check if the configuration $c^{(\iota,\iota')}$ has total size smaller than $1$, and if so, we check if its corresponding constraint \eqref{cons2} is violated.  
 
 If we found such a violating constraint, we are done. If we have found a constraint from the family \eqref{cons3} that is violated using the earlier cases, we are also done. Otherwise, we argue next that all constraints of the form \eqref{cons2} are satisfied.  To see this last claim, assume by contradiction that there is a configuration $c'\in \C_1$ with $s(c')>0$ whose corresponding constraint \eqref{cons2} is violated.  Let $(\iota,\iota')$ be the pair for which $c'$ is a feasible solution to $IP^3(\iota,\iota')$.  Since $c^{(\iota,\iota')}$ is optimal for this integer program, we have that $s(c^{(\iota,\iota')}) \leq s(c')$, and so $\iota + s(c^{(\iota,\iota')}) \cdot z_1+z_2 \leq \iota+ s(c')\cdot z_1+z_2$. Since the constraint corresponding to $c'$ is violated, the last term is smaller than $1$. Therefore, if $s(c^{(\iota,\iota')})>0$, then the constraint corresponding to configuration $c^{(\iota,\iota')}$ is also violated.  If however, $s(c^{(\iota,\iota')})=0$, then using $\iota<1$, we get that the constraint \eqref{cons3} corresponding to  $c^{(\iota,\iota')}$ is violated. This contradicts the assumption, so the claim holds.

\bibliographystyle{abbrv}

\end{document}